\documentclass[prd,twocolumn,aps,letterpaper,amsmath,amssymb,preprintnumbers,showpacs,superscriptaddress,floatfix,longbibliography,nofootinbib]{revtex4-2}
\usepackage{array}
\usepackage{graphicx}
\usepackage[export]{adjustbox}
\usepackage[caption=false]{subfig}
\usepackage{tikz}
\usepackage{mathrsfs}
\usetikzlibrary{tikzmark}
\usetikzlibrary{calc}
\allowdisplaybreaks 
\usepackage[hypertexnames=true]{hyperref}
\usepackage[capitalize]{cleveref}
\hypersetup{
    colorlinks=true,       
    linkcolor=blue,          
    citecolor=blue,        
    filecolor=blue,      
    urlcolor=blue           
}

\usepackage{slashed}
\usepackage{mathtools}
\usepackage{listings}
\usepackage{pgf}
\usepackage{fancyvrb}
\usepackage{longtable}
\usepackage{bm}
\usepackage{tikz-cd} 
\usepackage{amsthm}

\newtheorem{lemma}{Lemma}

\DeclareMathOperator{\imag}{Im}
\DeclareMathOperator{\real}{Re}

\DeclareMathOperator{\Tr}{Tr}
\DeclareMathOperator{\sgn}{sgn}
\newcommand{\Z}[0]{\ensuremath{\mathbb{Z}}}
\newcommand{\R}[0]{\ensuremath{\mathbb{R}}}
\newcommand{\C}[0]{\ensuremath{\mathbb{C}}}

\newcommand{\disk}[0]{\ensuremath{\mathbb{D}}}

\newcommand{\bra}[1]{\ensuremath{\langle #1 |}}   
\newcommand{\ket}[1]{\ensuremath{| #1 \rangle}}   

\newcommand{\avg}[1]{\ensuremath{\langle #1 \rangle}}
\newcommand{\abs}[1]{\ensuremath{| #1|}}

\definecolor{darkgreen}{rgb}{0.0, 0.545098, 0.0}

\begin{document}
\title{Hadronic Structure, Conformal Maps, and Analytic Continuation}

\author{Thomas~Bergamaschi}
\affiliation{Center for Theoretical Physics, Massachusetts Institute of Technology, Cambridge, MA 02139, USA}
\author{William~I.~Jay}
\email{willjay@mit.edu}
\affiliation{Center for Theoretical Physics, Massachusetts Institute of Technology, Cambridge, MA 02139, USA}
\author{Patrick~R.~Oare}
\email{poare@mit.edu}
\affiliation{Center for Theoretical Physics, Massachusetts Institute of Technology, Cambridge, MA 02139, USA}
\preprint{MIT-CTP/5563}

\date{\today}

\begin{abstract}
    We present a method for analytic continuation of retarded Green functions, including Euclidean Green functions computed using lattice QCD.
    The method is based on conformal maps and construction of an interpolation function which is analytic in the upper half plane.
    A novel aspect of our treatment is rigorous bounding of systematic uncertainties, which are handled by constructing the full space of interpolating functions (at each point in the upper half-plane) consistent with the given Euclidean data and the constraints of analyticity.
    The resulting Green function in the upper half-plane has an appealing interpretation as a smeared spectral function.
\end{abstract}

\maketitle

\section{Introduction\label{sec:introduction}}
Current-current correlation functions in quantum chromodynamics (QCD) encode fundamental features of hadronic structure.
For example, the hadronic vacuum polarization tensor
is defined as the vacuum expectation value of the commutator of the electromagnetic currents~\cite{Aoyama:2020ynm,Bernecker:2011gh}:
\begin{align}
\rho_{\mu\nu}(q)
&= \frac{1}{2\pi}\int d^4x\, e^{iq\cdot x} \bra{\emptyset} [j_\mu^{\rm EM}(x), j_\nu^{\rm EM}(0)]\ket{\emptyset},
\end{align}
with 
$\rho_{\mu\nu}(q)
= (q_\mu q_\nu - q^2 g_{\mu\nu})\rho(q^2)$, where $\rho(q^2)$ is the spectral density.
The spectral density is related to the experimentally measured $R$-ratio,
\begin{align}
    \rho(s) = \frac{R(s)}{12\pi^2} &&
    R(s) = \frac{\sigma(e^+e^- \to {\rm hadrons})}{4\pi \alpha^2/3s},
    \label{eq:Rratio_def}
\end{align}
where the denominator is the tree-level QED cross section for $e^+e^-\to\mu^+\mu^-$ in the massless limit ($m_\mu^2 \ll s$).
The experimental data for $R(s)$, reproduced in  \cref{fig:R} from the compilation of Ref.~\cite{Davier:2019can},
show a rich resonance structure, with prominent peaks near the masses of vector resonances.
These resonant peaks are the hadronic structure of the vacuum polarization.

\begin{figure}[t]
    \centering
    \includegraphics[width=0.48\textwidth]{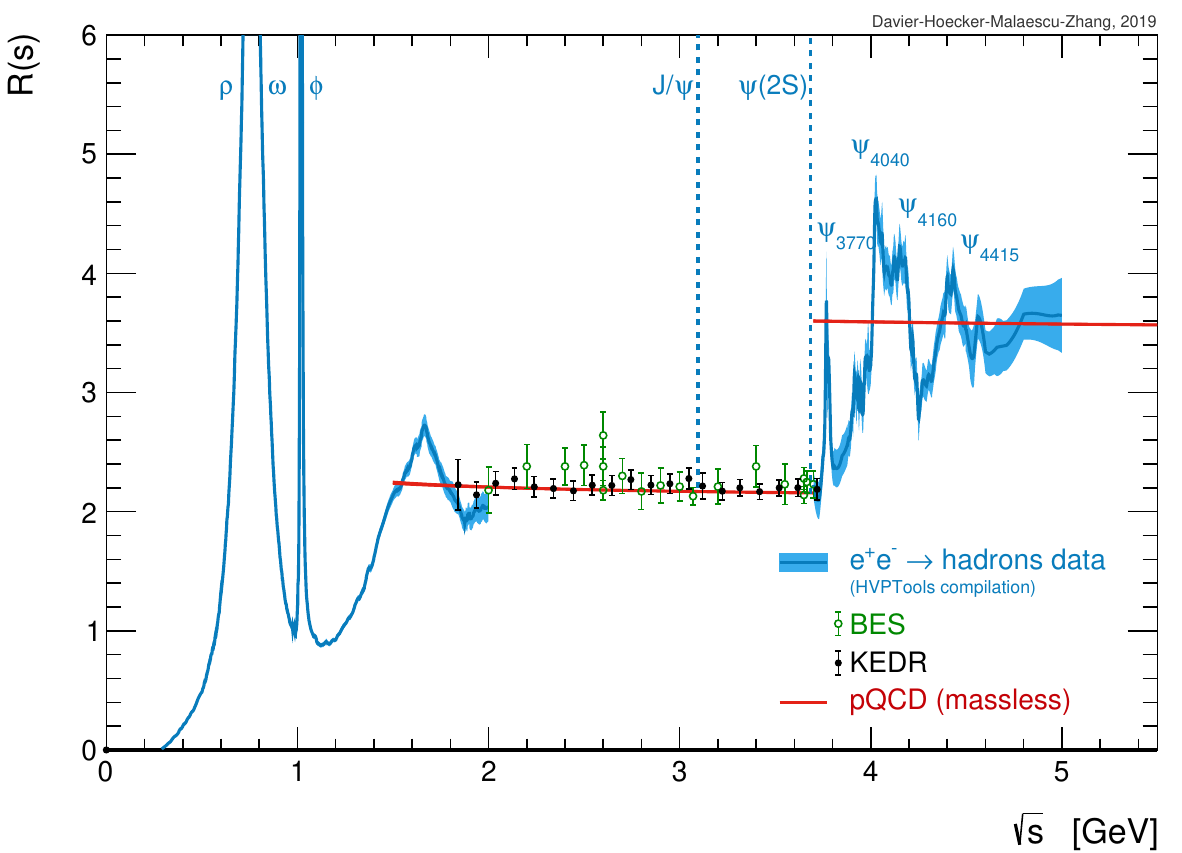}
    \caption{
        Compilation of experimental data for $R(s)$, reproduced with permission from Ref.~\cite{Davier:2019can}.
    }
    \label{fig:R}
\end{figure}

The hadronic tensor of unpolarized inclusive electron-proton scattering has a similar form~\cite{Manohar:1992tz}, 
\begin{align}
W_{\mu\nu}(p,q)
= \int \frac{d^4x}{4\pi} e^{iq\cdot x}
\bra{p} [j_\mu^{\rm EM}(x), j_\nu^{\rm EM}(0)]\ket{p},\label{eq:hadronic_tensor_EM}
\end{align}
where the external states correspond to a proton with four-momentum $p$ and $q=k-k'$ is the momentum transfer between the initial and final electrons with momenta $k$ and $k'$, respectively.
The Lorentz covariant decomposition of the hadronic tensor is given by
\begin{align}
    \begin{split}
    W_{\mu\nu}
    &= F_1\left(-g_{\mu\nu} + \frac{q_\mu q_\nu}{q^2} \right)\\
    &+ \frac{F_2}{p\cdot q}
    \left(
    p_\mu - \frac{p\cdot q q_\mu}{q^2}
    \right)
    \left(
    p_\nu - \frac{p\cdot q q_\nu}{q^2} 
    \right),
    \end{split}
\end{align}
where $F_1$ and $F_2$ are so-called structure functions.
Similar to the case of the hadronic vacuum polarization above, $F_1$ and $F_2$ can also be interpreted as spectral densities.
Moreover, as shown in \cref{fig:F2} for $F_2$, the experimentally measured structure functions exhibit conspicuous resonant peaks.
These structures encode the non-perturbative response of the proton to electromagnetic probes.
To date, essentially no first-principles understanding (e.g., from lattice QCD) of the structure functions exists in the resonance region.
For neutrino-nucleon scattering at similar energies, the analogue of \cref{eq:hadronic_tensor_EM} arises with flavor-changing vector and axial currents.
Compared to the electromagnetic case, the axial structure functions are especially poorly known.
Improved understanding of the axial structure functions would have important consequences for upcoming neutrino experiments like DUNE.

\begin{figure}[t]
    \centering
    \includegraphics[width=0.48\textwidth]{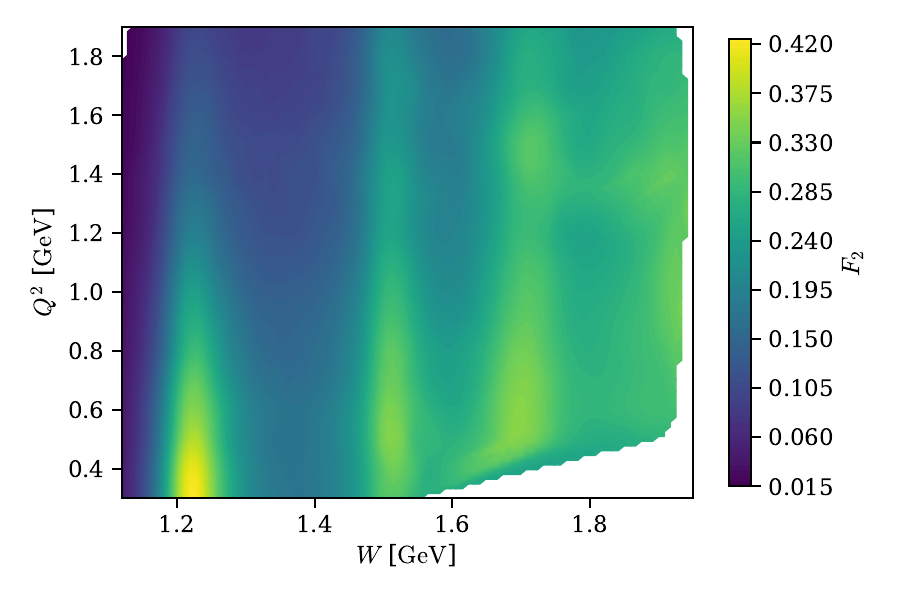}
    \caption{
        Interpolation of experimental data, taken from the CLAS Physics Database~\cite{Chesnokov:2022gjb}, for the structure function $F_2$ in the resonance region as a function of the momentum transfer $Q^2$ and the invariant mass $W$ of the hadronic system.
        Large resonant contributions are clearly visible states including $\Delta(1232)$ $N(1520)$, $N(1535)$, and $N(1720)$.
        In Ref.~\cite{Blin:2021twt}, at least a dozen open channels were used to model the experimental data.
    }
    \label{fig:F2}
\end{figure}

Euclidean-time analogues of 
\cref{eq:Rratio_def,eq:hadronic_tensor_EM} are calculable using lattice QCD.
Extracting a spectral density from a zero-temperature Euclidean correlation function $\mathscr{G}_E(\tau)$ requires inverting a Laplace transform,
\begin{align}
    \mathscr{G}_E(\tau) = \int_0^\infty d\omega\, e^{-\omega \tau}  \rho(\omega).\label{eq:laplace}
\end{align}
Lattice QCD calculations provide $\mathscr{G}_E(\tau)$ at a discrete set of Euclidean times.
Without further assumptions, this problem is famously ill posed.
Efforts to solve \cref{eq:laplace} have typically focused on the Laplace transform's structure as a linear integral equation, approximated by the discretized linear system
$\mathscr{G}_E(\tau_i) = K_{ij} \rho(\omega_j)$
where $K_{ij} = \int d\omega\, e^{-\tau_i\omega_j}$ and summation is implied.
Suitably regularized linear methods are then employed to extract the spectral density.
Much recent interest in the lattice community has been generated by the inversion algorithm of Ref.~\cite{Pijpers:1992}, independently rediscovered and applied to lattice gauge theory for the first time in Ref.~\cite{Hansen:2019idp}.
The general inverse problem has seen many approaches a broad literature spanning many fields~\cite{PhysRevB.96.035147,PhysRevE.95.061302,Itou:2020azb,doi:10.7566/JPSJ.89.012001,PhysRevB.101.035144,PhysRevB.98.035104,Shi:2022yqw,Chen:2021giw,Kades:2019wtd,Zhou:2023pti,Lechien:2022ieg,Wang:2021jou,
10.1093/gji/ggz520,DelDebbio:2021whr,Candido:2023nnb,Horak:2021syv,Pawlowski:2022zhh,Horak:2023xfb}; existing approaches used in the lattice community have recently been reviewed in Refs.
~\cite{Rothkopf:2022fyo,Rothkopf:2022ctl,Bulava:2023mjc}.

\Cref{eq:laplace} amounts to analytic continuation, a connection which is particularly clear in frequency space (see \cref{sec:finite_volume_green_functions} below),
\begin{align}
    G(i \omega_\ell) = \int_0^\beta d\tau e^{i\omega_\ell \tau} \mathscr{G}_E(\tau) ,
\end{align}
where the retarded Green function on the left-hand side is evaluated on the positive imaginary axis.
Meanwhile, the spectral density is related to its imaginary part evaluated on the positive real axis: $\rho(\omega)=\frac{1}{\pi} \imag G(\omega)$.
Therefore, the task of analytic continuation is to compute the behavior of $G(z)$ on the real line given finite data on the positive imaginary axis.

In the present work, we propose a method for solving this problem which differs from familiar approaches in three important regards.
First, the method is inherently nonlinear, based on special properties of certain conformal maps.
Second, the method works by constraining the behavior of $G(\omega + i \epsilon)$ directly in the upper half-plane, with the spectral density arising in the limit $\epsilon \to 0^+$.
Third, the method explicitly constructs, at each point in the upper half-plane, the full space of analytic functions consistent with the given Euclidean data.
Intuitively, one anticipates that this space ought to be small for points near the Euclidean data on the positive imaginary axis and large for points on the real line, where the problem becomes ill posed.
At each point $z\in\C^+$, this bounding space gives a rigorous bound on the systematic uncertainty associated with the analytic continuation.

The fact that lattice QCD calculations occur in a finite spatial volume raises important conceptual questions.
Ref.~\cite{Hansen:2017mnd} has convincingly argued that infinite-volume spectral densities may, as a matter of principle, be recovered from finite-volume calculations via the ordered limiting procedure
\begin{align}
\rho(\omega) 
    &=\lim_{\epsilon \to 0}
    \lim_{L \to \infty}
    \int d\omega'\, \delta_\epsilon(\omega, \omega')
    \rho_L(\omega')\\
    &\equiv\lim_{\epsilon \to 0}
    \lim_{L \to \infty}
    \rho^\epsilon_L(\omega) \label{eq:ordered_limit},
\end{align}
where $\rho_L(\omega)$ is a finite-volume spectral density, $\rho^\epsilon_L(\omega)$ is a smeared spectral density, and $\delta_\epsilon(\omega, \omega')$ is a smearing kernel.\footnote{In practical lattice QCD calculations, it may be advantageous to take the infinite-volume limit directly at the level of Euclidean-time correlation functions, with all other simulation parameters held fixed.
}
\footnote{For brevity, we will usually suppress the volume dependence in what follows.}
An important practical upshot is that computing smeared spectral densities is typically a well-posed, albeit numerically delicate, problem.
Below we show that analytic continuation at a finite distance $\epsilon$ above the real line has a natural interpretation in terms of a smeared spectral density.
Essentially the same point seems to have been observed in Ref.~\cite{Bulava:2019kbi}, where the usual $i\epsilon$ prescription for computing scattering amplitudes was replaced by a suitable smearing in a proposal to compute scattering amplitudes using lattice gauge theory.

Recently a similar set of ideas to the method we propose was applied to fermionic thermal Green functions in the context of condensed matter problems~\cite{PhysRevLett.126.056402,PhysRevB.104.165111}.
Bosonic Green functions are of particular interest in lattice gauge theory.
One approach to bosonic Green functions has been given in Ref.~\cite{Nogaki:2023mut}.
The present work gives an alternative treatment of bosonic Green functions.
Our rigorous treatment of the uncertainty in the analytic continuation is another novel feature of the present work.

The rest of this article is organized as follows.
\Cref{sec:thermal_green_functions} reviews some well-known analytic properties of thermal Green functions in the upper half-plane.
\Cref{sec:finite_volume_green_functions} specializes the results of \cref{sec:thermal_green_functions} to the case of a finite spatial volume, where the spectrum is discrete.
\Cref{sec:conformal_maps} transforms, by suitable conformal maps, the problem of analytic continuation from the upper half-plane to the open unit disk.
\Cref{sec:interpolation} constructs a rational-function approximation for $G(z)$ using the theory of Nevanlinna--Pick interpolation; \Cref{ssec:complete_algorithm} collects the technical pieces and summarizes the complete algorithm for the method.
\Cref{sec:numerical_GR} provides a recipe for computing $G(i\omega_\ell)$ given values for $\mathscr{G}_E(\tau)$, e.g., from a lattice QCD calculation.
\Cref{sec:examples} provides a series of numerical examples illustrating how the new method works in practice.
\Cref{sec:conclusions} provides some discussion of the results and presents our conclusions.

\section{Thermal Green functions in the complex plane\label{sec:thermal_green_functions}}

Consider a finite-temperature quantum field theory defined by an equilibrium density matrix
\begin{align}
\hat{\rho} = \frac{1}{\mathcal{Z}} e^{-\beta H},
\end{align}
with partition function $\mathcal{Z}$, Hamiltonian $H$, and inverse temperature $\beta$.\footnote{For completeness, the next two sections review standard definitions and relations between various correlation functions.
Aside from incidental comments, no novelty is claimed.
Similar material can be found with an emphasis on thermal properties, e.g., in Ref.~\cite{Meyer:2011gj}.
}
Expectation values are defined with respect to the density matrix as
$\avg{\mathcal{O}_1 \dots \mathcal{O}_n}
    \equiv \Tr \left\{ \hat{\rho}\, \mathcal{O}_1 \dots \mathcal{O}_n \right\}$.
Although many applications of interest occur at zero temperature, the formalism we present holds for arbitrary temperatures.
Moreover, as a matter of principle, lattice QCD calculations employ large but finite $\beta$.

Let $A$ be an operator.
We define the following correlation functions
\begin{align}
\mathscr{G}(t) &= \avg{A(t) A^\dagger(0)}, \\
\mathscr{G}_\pm(t) &= i\avg{ \{A(t), A^\dagger(0)\}_\pm} \label{eq:Gpm},
\end{align}
where the commutator $(-)$ arises for bosonic operators, while the anti-commutator $(+)$ is for fermionic operators.
The Euclidean Green function is defined via analytic continuation to the lower-half plane as
\begin{align}
    \mathscr{G}_E(\tau) \equiv \mathscr{G}(-i \tau),
    \label{eq:GE_definition}
\end{align}
where $\tau \in \R$ is the Euclidean time.
The retarded, or causal, correlator is defined as 
\begin{align}
G_{\pm}(\omega) &\equiv \int_0^\infty dt\, e^{i\omega t} \mathscr{G}_\pm(t).\label{eq:GR}
\end{align}
Nominally, $\omega$ is a real number.
When $\omega$ is replaced by a complex number, \cref{eq:GR} defines an analytic function $G_{\pm}(z)$ in the upper half plane, $\imag z > 0$.
The retarded correlator plays a key role in linear response theory~\cite{Meyer:2011gj}; its importance to the current discussion arises from its close connection to the Euclidean correlator.
The Fourier coefficients of the Euclidean correlator are defined as
\begin{align}
G_E^{(\ell)} \equiv \int_0^\beta d\tau\, e^{i \omega_\ell \tau} \mathscr{G}_E(\tau),
\label{eq:fourier_coefficients}
\end{align}
where the $\omega_\ell$ are the Matsubara frequencies, $2\ell\pi/\beta$ for bosons and $(2\ell+1)\pi/\beta$ for fermions, with $\ell \in \Z$.
In \cref{sec:finite_volume_green_functions} below, we will rederive the familiar result that 
\begin{align}
G_E^{(l)} = G_{\pm}(i\omega_\ell),\, \ell \neq 0.
\end{align}
In other words, the analytic continuation of the retarded correlator is the frequency-space Euclidean correlator.

Finally, the spectral density is defined as
\begin{align}
\rho_\pm(\omega)
    &= \frac{1}{2\pi i} \int_{-\infty}^{\infty} dt\, e^{i \omega t} \mathscr{G}_\pm(t)\\
    &= \frac{1}{2\pi i} \left( G_{\pm}(\omega) - G_{\pm}(\omega) ^*\right) \\
    &= \frac{1}{\pi} \imag G_{\pm}(\omega), \,\omega \in \R. \label{eq:rho_equals_ImGR}
\end{align}
where the second line follows using time-translation invariance
and the reality condition.
The final line gives the familiar result: for diagonal correlators, the spectral density is the imaginary part of the retarded Green function.

\begin{table}[]
    \begin{tabular}{c l l}
    \hline\hline
    Symbol                  & Description   & Definition \\
    \hline
    $\mathscr{G}(t)$        & Real-time Green function & \cref{eq:Gpm} \\
    $\mathscr{G}_E(\tau)$   & Euclidean-time Green function & \cref{eq:GE_definition}\\
    $G(\omega)$             & Retarded Green function, $\omega\in\R$ & \cref{eq:GR}\\
    $G(z)$                  & Retarded Green function, $z\in\C^+$ & \cref{eq:GR}\\
    $\mathcal{G}(z)$        & Retarded Green function, $z\in\disk$ & \cref{eq:GRplus_disk,eq:GRminus_disk}\\
    \hline\hline
    \end{tabular}
    \caption{
    The different Green functions appearing in this work.
    Bosonic and fermionic Green functions are distinuished by the sign of the commutator or anti-commutator, e.g., $G_\pm(z)$.
    }
    \label{table:green_functions}
\end{table}

\section{Finite-volume Green functions\label{sec:finite_volume_green_functions}}

The definitions of the preceding section were generic, in the sense that they made no particular assumption about the dynamics or spectrum of the theory.
We now specialize to the case of a thermal field theory in a finite spatial volume $V=L^3$, for which the spectrum is discrete.

Inserting complete sets of states in \cref{eq:Gpm} gives
\begin{align}
\mathscr{G}_\pm(t) &= \frac{i}{\mathcal{Z}} \sum_{n,m} e^{-iE_{nm}t} |A_{mn}|^2 \left( e^{-\beta E_m} \pm e^{-\beta E_n} \right),
\label{eq:Gpm_spectral_decomp}
\end{align}
where $E_{mn} \equiv E_m - E_n$ and $A_{mn} \equiv \bra{m}A\ket{n}$. 
Similarly, the retarded correlator in the upper half-plane $z\in \C^+$ becomes
\begin{align}
G_{\pm}(z) = \frac{1}{\mathcal{Z}} \sum_{n,m} |A_{mn}|^2 \left( e^{-\beta E_m} \pm e^{-\beta E_n} \right) \frac{-1}{z - E_{nm}},\label{eq:GR_finite_volume}
\end{align}
where we have used
\begin{align}
\int_0^\infty dt\, e^{i z t} e^{-i E_{nm} t} = \frac{i}{z - E_{nm}}
\end{align}
for $E_{mn} \in \R$ and $z \in \C^+$.
This confirms the statement from above that $G_{\pm}(z)$ is analytic in the upper half plane.
To extract the spectral density, we evaluate the pole in the upper half-plane at $z=\omega + i \epsilon$ with $\omega, \epsilon \in \R$.
It follows that 
\begin{align}
    \frac{1}{\pi}\imag \frac{-1}{(\omega + i \epsilon) - E_{nm}}
    &= \frac{1}{\pi}\frac{\epsilon}{(\omega - E_{nm})^2 + \epsilon^2}\\
    &\equiv \delta_\epsilon(\omega - E_{nm}),
\end{align}
where $\delta_\epsilon(x)$ is the Poisson kernel, which approaches the Dirac delta function in the usual distributional sense:
\begin{align}
    \lim_{\epsilon\to 0^+} \delta_\epsilon(x) = \delta(x).
\end{align}
Combining this result with \cref{eq:rho_equals_ImGR,eq:GR_finite_volume} confirms that the spectral density is a discrete sum of delta functions:
\begin{align}
\rho_\pm(\omega) &= \frac{1}{\mathcal{Z}} \sum_{n,m} |A_{mn}|^2
\left( e^{-\beta E_m} \pm e^{-\beta E_n} \right) \delta(\omega - E_{nm})\label{eq:rho_fv}\\
&\stackrel{\beta\to\infty}{=}
\sum_{n} |A_{0n}|^2
\big(
\delta(\omega - E_{n})
\pm \delta(\omega + E_{n})
\big).\label{eq:rho_fv_zerotemp}
\end{align}
For all temperatures, the spectral density is evidently an even function for fermions and an odd function for bosons.
The second line follows in the zero-temperature limit.
At zero temperature, and for a given fixed set of overlap factors $|A_{0n}|$ and energies $E_n$, 
the only difference between $\rho_{+}(\omega)$ and $\rho_{-}(\omega)$ is the relative minus sign on the negative real line.

Using the Poisson kernel, we define a ``smeared" spectral density via
\begin{align}
\rho^\epsilon_{\pm}(\omega) = 
\int d\omega^\prime \delta_\epsilon(\omega - \omega^\prime) \rho_\pm(\omega^\prime).
\end{align}
Substituting the explicit form for the spectral density in \cref{eq:rho_fv} 
then gives a useful generalization of \cref{eq:rho_equals_ImGR}:
\begin{align}
\rho^\epsilon_{\pm}(\omega)
= \frac{1}{\pi} \imag G_{\pm}(\omega + i \epsilon),
\label{eq:rhoeps_equals_ImGR}
\end{align}
which is valid for arbitrary $z=\omega+i\epsilon$ in the upper half-plane.
It is worth emphasizing that \cref{eq:rhoeps_equals_ImGR} does \emph{not} require $\epsilon$ to be small.
This result is noteworthy because it says that the smeared finite-volume spectral function, defined with the Poisson kernel, is the analytic continuation of the retarded Green function $G_{\pm}(\omega)$.
In particular, this formula establishes the close connection between the present work and recent work in Refs.~\cite{Hansen:2019idp,Bulava:2019kbi,Hansen:2017mnd}.
This smearing is also in the spirit of the classic work of Ref.~\cite{Poggio:1975af}.

It remains to relate the Euclidean Green function to 
the retarded Green function and the spectral density.
From \cref{eq:GE_definition,eq:Gpm_spectral_decomp}, it immediately follows that the spectral decomposition of $\mathscr{G}_E(\tau)$ is given by
\begin{align}
\mathscr{G}_E(\tau) = \frac{1}{\mathcal{Z}} \sum_{n,m} e^{-\beta E_m} e^{-E_{nm}\tau} \abs{A_{mn}}^2.
\end{align}
From the definition of Fourier coefficients $G_E^{(l)}$ in \cref{eq:fourier_coefficients}, 
it follows that $G_E^{(l)} = G_{\pm}(i \omega_l)$ at the appropriate bosonic or fermionic Matsubara frequencies.
Finally, comparison with \cref{eq:rho_fv} delivers the relation between $\mathscr{G}_E(t)$ and the spectral density
\begin{align}
    \mathscr{G}_E(\tau)
    &= \int_0^\infty d\omega\,\rho_\pm(\omega) \left[
    \frac{e^{-\omega \tau} + e^{-\omega(\beta-\tau)}}
    {1\pm e^{-\omega \beta}}
    \right]
    \label{eq:laplace_finite_temp}
    \\
    &\stackrel{\beta \to \infty}{=}
    \int_0^\infty d\omega\, \rho_\pm(\omega) e^{-\omega \tau}.\label{eq:laplace_zero_temp}
\end{align}
As anticipated in \cref{sec:introduction}, the second line shows that the zero-temperature Euclidean Green function is the Laplace transform of the spectral density.
For large but finite $\beta$, the backward-propagating contribution $e^{-\omega(\beta-\tau)}$  should be retained, while the denominator can be safely neglected.
For generic $\beta$, the denominator is related to the Fermi--Dirac ($+$) or Bose--Einstein ($-$) distribution.
The notation in \cref{eq:laplace_zero_temp}  echoes \cref{eq:rho_fv_zerotemp}, 
where zero-temperature fermionic and bosonic spectral densities agree for $\omega > 0$.

\section{Transforming Green functions to the unit disk\label{sec:conformal_maps}}

As discussed above, the retarded Green function is analytic in the upper half plane.
In a finite volume, \cref{eq:GR_finite_volume} shows that it consists of a sum of pairs of poles at $\omega = \pm E_{nm}$.
The associated spectral densities in \cref{eq:rho_fv} have definite parity: $\rho_+(\omega)$ is an even function, while $\rho_-(\omega)$ is an odd function.

The even parity of the fermionic spectral density implies that the retarded correlator is a function $G_{+}: \C^+ \to \C^+$, a fact also manifest directly in \cref{eq:GR_finite_volume}.
In other words, the retarded correlator satisfies a positivity condition in the upper half-plane: $\imag G_{+}(z) > 0$ for all $z \in \C^+$.
Functions satisfying this condition are known as Nevanlinna functions and have been studied extensively in complex analysis.
This property of fermionic Green functions was a key insight in Refs.~\cite{PhysRevLett.126.056402,PhysRevB.104.165111}.
It will prove useful to define an associated function on the unit disk $\mathbb{D}$.
To this end, and as in Refs.~\cite{PhysRevLett.126.056402,PhysRevB.104.165111}, the Cayley transform is defined as the map $C: \C^+ \to \disk$
\begin{align}
    C(z) = \frac{z-i}{z+i}
    &&C^{-1}(\zeta) = -i \left(\frac{\zeta+1}{\zeta-1}\right).
    \label{eq:cayley}
\end{align}
Here and below, we reserve the variable $\zeta$ for complex numbers in the unit disk; the variable $z$ will be used for complex numbers in the upper half-plane.
The Cayley transform is summarized visually in \cref{fig:cayley}.
Composition with the Cayley transform maps the domain and co-domain to $\disk$, yielding the desired function $\mathcal{G}_{+}:\disk \to \disk$,
\begin{align}
    \mathcal{G}_{+}(\zeta)
    = (C \circ G_{+} \circ C^{-1})(\zeta).
    \label{eq:GRplus_disk}
\end{align}

\begin{figure}[t]
    \centering
    \includegraphics[width=0.48\textwidth]{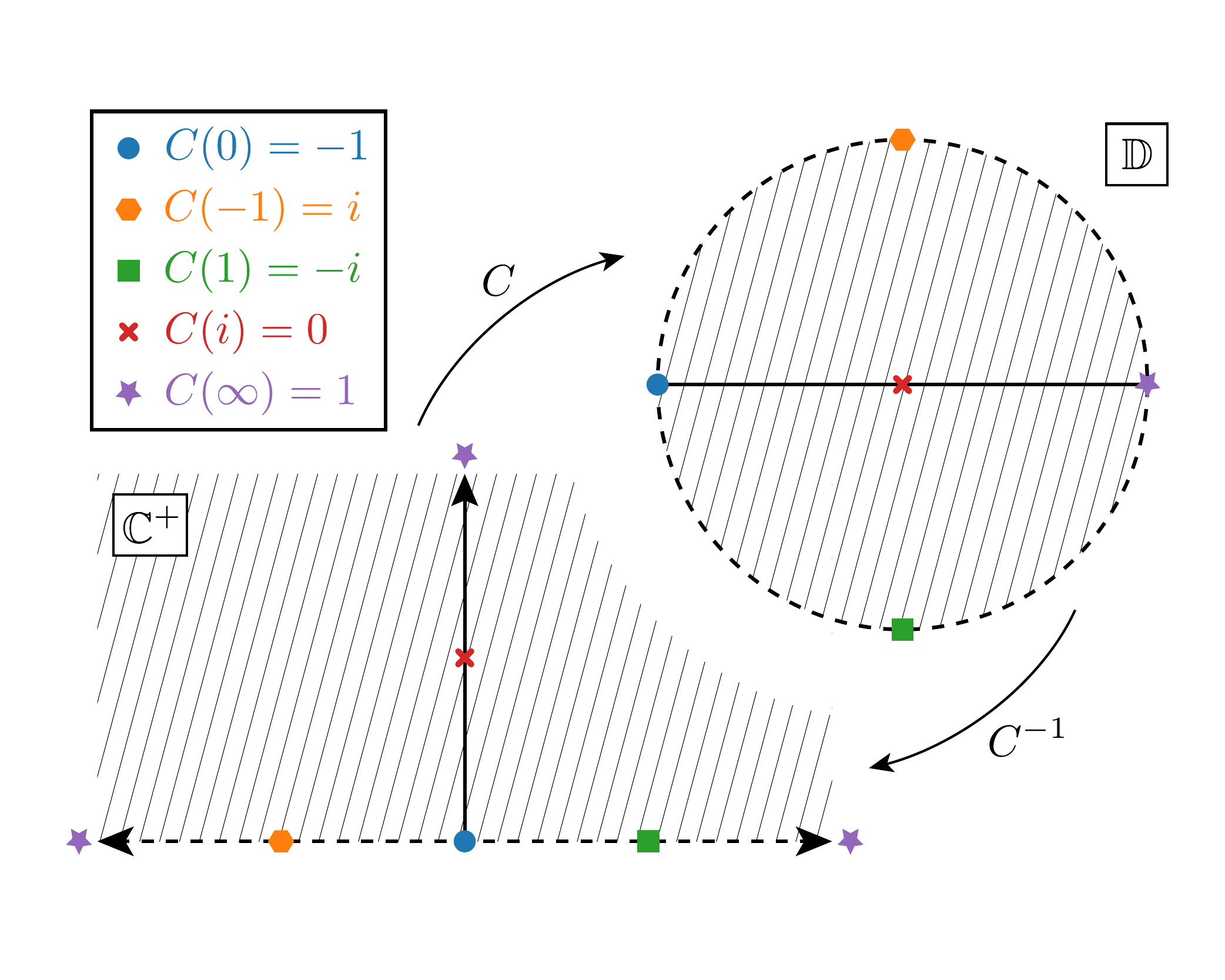}
    \caption{
       Diagramatic representation of the Cayley transform $C:\C^+\to\disk$ and the inverse transform $C^{-1}:\disk\to\C^+$.
       The real line is mapped to boundary of the unit disk, while the upper half-plane is mapped to its interior.
    }
    \label{fig:cayley}
\end{figure}

\begin{figure}[t]
    \centering
    \includegraphics[width=0.48\textwidth]{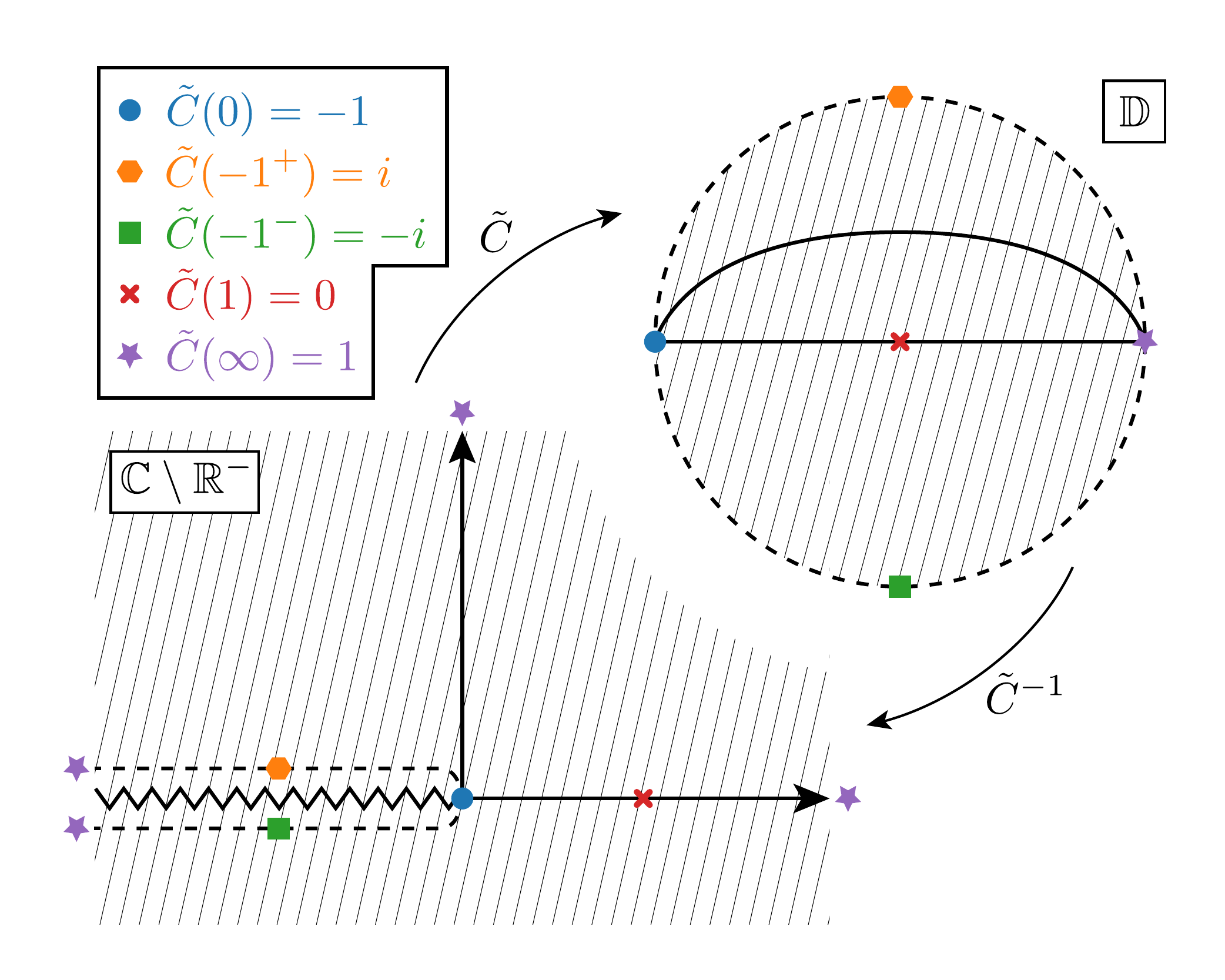}
    \caption{
       Diagrammatic representation of the transform $\widetilde{C}:\C\setminus R^- \to \disk$ and the inverse transform $\widetilde{C}^{-1}:\disk \to \C\setminus R^-$. The branch cut on the negative real line is mapped to the boundary of the unit disk. The positive real and imaginary axes are mapped to the indicated solid lines in the interior of the disk.
    }
    \label{fig:ctilde}
\end{figure}

Since the bosonic spectral density is an odd function, $G_{-}$ is not a Nevanlinna function.\footnote{The method described here for bosonic Green functions differs from the treatment in Ref.~\cite{Nogaki:2023mut}. See \cref{ssec:complete_algorithm} for a discussion of differences.}
Indeed, \cref{eq:GR_finite_volume} shows that $\imag G_{-}$ is positive in the first quadrant and negative in the second quadrant.
By symmetry, $\imag G_{-}$ only vanishes in the upper half plane along the positive imaginary axis, where the real part is strictly positive, $\real G_{-}(iy) > 0$ for $y>0$.
In other words, $G_{-}: \C^+ \to \C \setminus \R^-$, where $\R^-$ is the negative real line.
This motivates defining a modified transform $\widetilde{C}: \C\setminus\R^- \to \disk$ via
\begin{align}
\widetilde{C}(z)=\frac{\sqrt{z}-1}{\sqrt{z}+1},
&&\widetilde{C}^{-1}(\zeta) = \left(\frac{1+\zeta}{1-\zeta}\right)^2,
\label{eq:Ctilde}
\end{align}
with the usual branch cut along the negative real line.
The transformation $\tilde{C}$ is summarized visually in~\cref{fig:ctilde}.
Similar to the fermionic case, composition with $\widetilde{C}$ and the Cayley transorm then yields an associated function on the unit disk $\mathcal{G}_{-}:\disk\to\disk$:
\begin{align}
\mathcal{G}_{-}(\zeta) = (\widetilde{C}\circ G_{R-} \circ C^{-1})(\zeta)\label{eq:GRminus_disk}
\end{align}

The functions $\mathcal{G}_{\pm}(\zeta)$ are members of the Schur class $\mathscr{S}$,
\begin{align}
    \mathscr{S} = \{f: \disk \to \overline{\disk}: f~\text{is analytic}\}. \label{eq:schur_class}
\end{align}
The Schur class consists of non-constant analytic functions $f:\disk\to\disk$ as well as constant functions taking values in $\overline{\disk}$.
Viewed as elements of the Schur class, the functions $\mathcal{G}_\pm(\zeta)$ play a key role in the construction of interpolating functions below.

\section{Nevanlinna  interpolation\label{sec:interpolation}}

The basic idea to constrain an analytic function starting from an interpolation problem in the complex plane dates back more than a century, to work from Pick~\cite{Pick1915} and Nevalinna~\cite{Nevalinna1919,Nevalinna1929}.
Most of the fundamental results needed for the present work were recently reviewed with modern terminology and notation in recent lectures by Nicolau~\cite{Nicolau2016}, whose discussion we follow closely.
Additional technical details and context are accessibly presented in Refs.~\cite{BlaschkeBook,PickInterpolationBook}.
To keep the discussion self contained, we reproduce proofs for most of the necessary mathematical results.
Although these results are well known to experts, we think it is valuable to collect them here with a consistent notation.
The idea to use Nevanlinna interpolation for fermionic correlators was recently presented in Refs.~\cite{PhysRevLett.126.056402,PhysRevB.104.165111}.
To our knowledge, Refs.~\cite{PhysRevLett.126.056402,PhysRevB.104.165111} were the first to apply the recursive approach to Nevanlinna interpolation, which we also follow, to problems in field theory.

\subsection{An Interpolating Function on the Disk}

In the preceding sections, we showed how the retarded Green functions $G_{\pm}(z)$ can be transformed into functions $\mathcal{G}_{\pm}(\zeta): \disk \to \disk$.
We suppose that the Matsubara frequencies and Euclidean data have been mapped to the sets
\begin{align}
    \{ i\omega_\ell\} \mapsto \{ \zeta_\ell \} \subset \disk\\
    \{ G_{\pm}(i\omega_\ell) \} \mapsto \{ w_\ell \} \subset \disk
\end{align}
using \cref{eq:cayley,eq:Ctilde}.
\Cref{fig:eval_axis} summarizes the setup for the Matsubara frequencies, which are mapped to the interior of the unit disk.
The task now is to construct a rational function $f(\zeta): \disk\to\disk$ which interpolates these points, i.e., satisfying $f(\zeta_l) = w_l$.
By construction, $f(\zeta)$ will be analytic on the unit disk
and an element of the Hardy space $H^\infty$ (see, e.g., Ref.~\cite{nikolski_2019} for a discussion of Hardy spaces).
To begin, it is useful to establish some notation.
First, Blaschke factors are defined according to
\begin{align}
    b_a(\zeta) = \frac{\abs{a}}{a} \frac{a-\zeta}{1 - \bar{a}\zeta} 
    && a\in\mathbb D\setminus\{0\},
\end{align}
with $b_0(\zeta)\equiv \zeta = \mathrm{id}_{\mathbb D}$. Within particle physics, Blaschke factors play a familiar role in the conformal maps of the $z$-expansion used, for instance in the study of semileptonic decays of mesons~\cite{Caprini:1997mu,Boyd:1994tt,Boyd:1995cf,Boyd:1995sq,Boyd:1997kz,Grinstein:2017nlq}.
Blaschke factors enjoy several useful properties relevant for the present discussion.
First, for $\abs{a} < 1$, Blaschke factors are automorphisms of the unit disk.
Second, they clearly satisfy $b_a(a)$ = 0.
Combining these properties with the maximum modulus principle means that any analytic function $g:\disk\to\disk$ with a zero at $a\in\disk$ can be written in the factorized form $g(\zeta) = b_a(\zeta) \tilde{g}(\zeta)$ where $\tilde{g}(\zeta):\disk\to\disk$ is another analytic function~\cite{BlaschkeBook,PickInterpolationBook}.
Finally, the following matrix notation for M{\"o}bius maps will prove very convenient,
\begin{align}
\begin{pmatrix}
    a(\zeta) & b(\zeta) \\
    c(\zeta) & d(\zeta)
\end{pmatrix} h(\zeta)
\equiv 
\frac{a(\zeta) h(\zeta) + b(\zeta)}{c(\zeta) h(\zeta) + d(\zeta)},
\end{align}
where $h(\zeta)$ is a function.
One of the chief utilities of this notation is that function composition corresponds to matrix multiplication, with the matrix inverse coinciding to the function inverse.

\begin{figure}[t]
    \centering
    \includegraphics[width=0.48\textwidth]{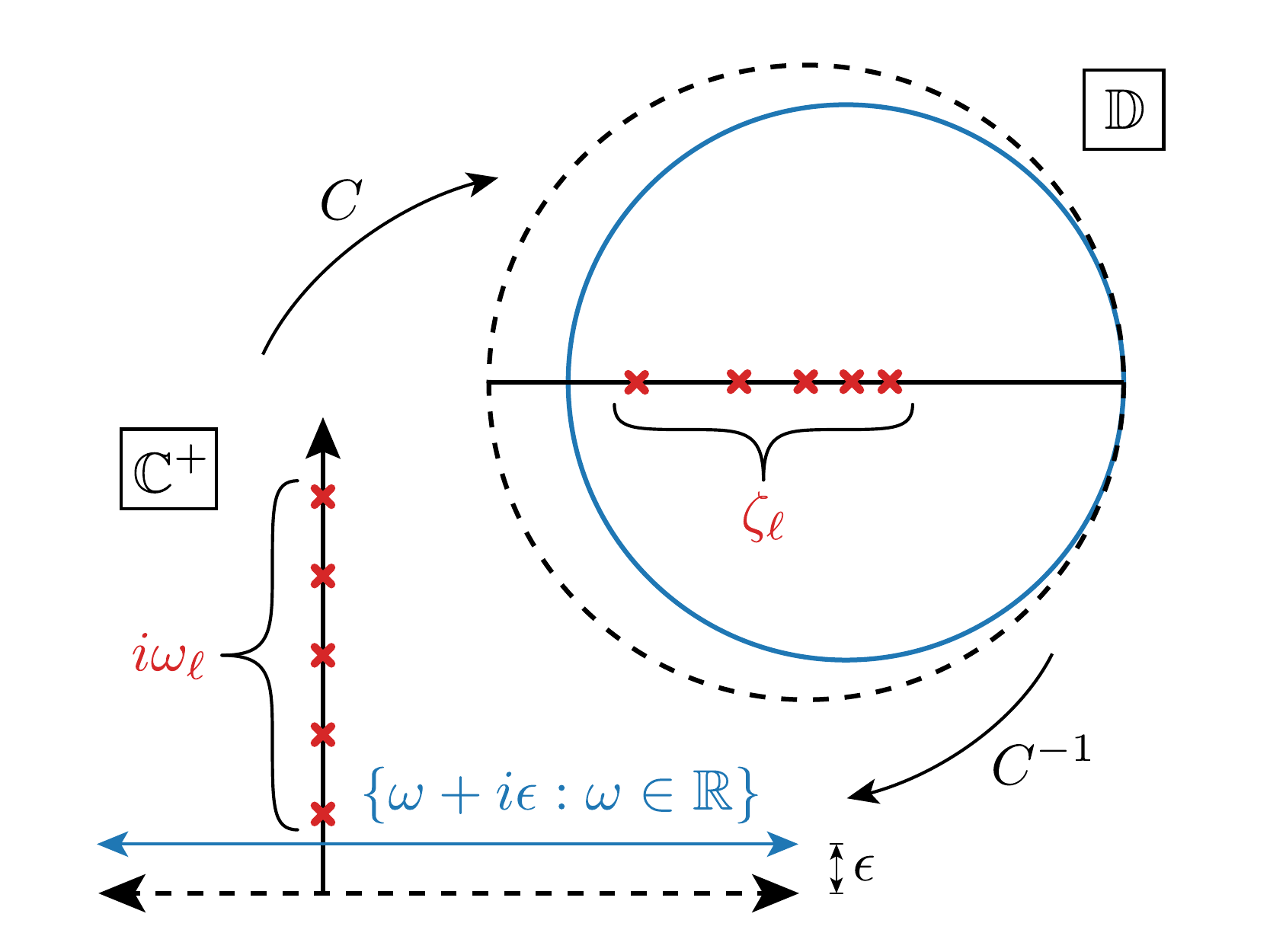}
    \caption{
       Mapping the Matsubara frequencies from the upper half-plane to the unit disk using the Cayley transform defined in \cref{eq:cayley}.
       The blue line a distance $\epsilon$ above the real line is mapped to the blue circle in $\disk$.
       We emphasize that the same mapping to the disk has been considered, for fermionic correlators, in Refs~\cite{PhysRevLett.126.056402,PhysRevB.104.165111} (cf. Fig 1 in Ref.~\cite{PhysRevLett.126.056402}).
    }
    \label{fig:eval_axis}
\end{figure}

Let us now apply these ideas to the interpolation problem, proceeding by induction.
For the base case, imposing $f(\zeta_1) - w_1 =0$ allows one to write
\begin{align}
    \frac{f(\zeta) - w_1}{1 - \overline{w_1} f(\zeta)} = b_{\zeta_1}(\zeta) f_1(\zeta) \label{eq:nevanlinna_base_case}
\end{align}
for some $f_1(\zeta) \in H^\infty$, since the left-hand side is an analytic function from the disk onto itself with zero at $w_1$.
Solving for $f(\zeta)$ gives
\begin{align}
    f(\zeta)
    &= \frac{b_{\zeta_1}(\zeta) f_1(\zeta) + w_1}{1 + \overline{w_1} b_{\zeta_1}(\zeta) f_1(\zeta)}\\
    &= \frac{1}{\sqrt{1-\abs{w_1}^2}}
    \begin{pmatrix}
        b_{\zeta_1}(\zeta) & w_1 \\
        \overline{w}_1 b_{\zeta_1}(\zeta) & 1
    \end{pmatrix} f_1(\zeta)\\
    &\equiv U_1(\zeta) f_1(\zeta),
\end{align}
where the normalization of the matrix $U_1(\zeta)$ has been chosen for later convenience so that $\det U_1(\zeta) = b_{\zeta_1}(\zeta)$.
Imposing the interpolation condition $f(\zeta_2) = w_2$ provides a constraint on the function $f_1(\zeta)$.
Namely, \cref{eq:nevanlinna_base_case} says that
\begin{align}
f_1(\zeta_2) = \frac{1}{b_{\zeta_1}(\zeta_2)} \frac{w_2 - w_1}{1 - \overline{w}_1 w_2} \equiv w_2^{(1)}
\end{align}
Thus, a Blaschke factor may be again be factored out to give
\begin{align}
\frac{f_1(\zeta) - w_2^{(1)}}{1-\overline{w}_2^{(1)} f_1(\zeta)} = b_{z_2}(\zeta) f_2(\zeta),
\end{align}
in terms of some $f_2(\zeta) \in H^\infty$.
Solving for $f_1(\zeta)$ then gives
\begin{align}
    f_1(\zeta) = \frac{b_{\zeta_2}(\zeta) f_2(\zeta) + w_2^{(1)}}{1 + \overline{w}_2^{(1)} b_{\zeta_2}(\zeta) f_2(\zeta)} \equiv U_2(\zeta) f_2(\zeta),
\end{align}
with $U_2(\zeta)$ defined analogously to $U_1(\zeta)$.
With this result, the inductive step now becomes obvious. 
Using all $N$ points,
\begin{align}
    f(\zeta)
    &= U_1(\zeta) U_2(\zeta) \cdots U_N(\zeta) f_N(\zeta)\\
    &\equiv \begin{pmatrix}
        P_N(\zeta) & Q_N(\zeta) \\
        R_N(\zeta) & S_N(\zeta) \\
    \end{pmatrix} f_N(\zeta) \label{eq:interpolating_family}
\end{align}
with $f_N(\zeta) \in H^\infty$.
The functions $P_N$, $Q_N$, $R_N$, $S_N$ are call the Nevalinna coefficients.
The matrices $U_n(\zeta)$ are defined according to
\begin{align}
    U_n(\zeta) = \frac{1}{\sqrt{1-\bigl| w_n^{(n-1)} \bigr|^2 }}
    \begin{pmatrix}
        b_{\zeta_n}(\zeta) & w_n^{(n-1)} \\
        \overline{w}_n^{(n-1)} b_{\zeta_n}(\zeta) & 1
    \end{pmatrix}
\end{align}
where $w_m^{(n)}$ is the value of $f_n$ evaluated at the $m$th zero: $w_m^{(n)} \equiv f_n(\zeta_m)$.
An explicit formula for $w_m^{(n)}$ follows from imposing the interpolation condition
\begin{align}
f(\zeta_m) &= U_1(\zeta_m) ... U_n(\zeta_m) w_m^{(n)} = w_m\\
w_m^{(n)} &= U_n^{-1}(\zeta_m) U_{n-1}^{-1}(\zeta_m) \cdots U_1^{-1}(\zeta_m) w_m
\end{align}
Since $f_N(\zeta) \in H^\infty$ is an arbitrary function, \cref{eq:interpolating_family} actually parameterizes the full space of possible analytic continuations of the given finite data.
Since freedom generically exists to include more points in the interpolation, the analytic continuation is not unique.
Remarkably, the size of the solution space can quantified sharply.

To see this, consider space of solutions to the $N$-point interpolation problem evaluated at $\zeta$:
\begin{align}
    \Delta_N(\zeta) = \Bigl\{ f(\zeta): f\in H^\infty, f(\{\zeta_n\})=\{w_n\}\Bigr\}.
\end{align}
The size of this set determines the ambiguity in the analytic continuation at the point $\zeta$.
For fixed $\zeta \in \overline\disk$, the set can be written explicitly.
The idea is to view \cref{eq:interpolating_family} as a function of the undetermined $f_N(\zeta) \in \disk$.
From this viewpoint, it immediately follows that
\begin{align}
\Delta_N(\zeta) = \{ T_{N,\zeta}(w): w \in \disk \},
\end{align}
where $T_{N,\zeta}:\disk \to \disk$ is the function
\begin{align}
    T_{N,\zeta}(w)
    = \frac{P_N(\zeta) w + Q_N(\zeta)}{R_N(\zeta) w + S_N(\zeta)}.
    \label{eq:Tnz}
\end{align}
This parameterization makes it clear that $\Delta_N(\zeta)$ is a disk, since M{\"obius} transformation map circles to circles.
As proven in \cref{app:disk_bounds}, this function is \emph{into} for all $\abs{\zeta} < 1$.
Moreover, the function is \emph{onto} when $\abs{\zeta}=1$.
We can also compute the center and radius of the disk $\Delta_N(\zeta)$.
Since $T_{N,\zeta}(-S_N(\zeta)/R_N(\zeta))=\infty$, the reflection property of M{\"o}bius transformations implies that is $\Delta_N(\zeta)$ centered at
\begin{align}
    c_N(\zeta)
    &=\frac
    {P_N(\zeta) \overline{(-R_N(\zeta)/S_N(\zeta))} + Q_N(\zeta)}
    {R_N(\zeta) \overline{(-R_N(\zeta)/S_N(\zeta))} + S_N(\zeta)}\\
    &=\frac
    {Q_N(\zeta) \overline{S}_N(\zeta) - P_N(\zeta) \overline{R}_N(\zeta)}
    {\abs{S_N(\zeta)}^2 - \abs{R_N(\zeta)}^2}.
    \label{eq:wertevorrat_center}
\end{align}
The radius of $\Delta_N(\zeta)$ can be computed by evaluating the distance between the center and its boundary.
A brief calculation shows that the radius is given by
\begin{align}
    r_N(\zeta) = \frac{\abs{B_N(\zeta)}}{\abs{S_N(\zeta)}^2 - \abs{R_N(\zeta)}^2}, \label{eq:wertevorrat_boundary}
\end{align}
where $B_N(\zeta)=P_N(\zeta) S_N(\zeta) - Q_N(\zeta) R_N(\zeta)$.
Collecting all the results, we see that $\Delta_N(\zeta) \subseteq \disk$ is a disk of radius $r_N(\zeta)$ centered at $c_N(\zeta)$.
Following the original work Nevanlinna~\cite{Nevalinna1919,Nevalinna1929}, this disk is sometimes known in the complex analysis literature as the \emph{Wertevorrat}.\footnote{
Consider a function $f:A\to B$ with co-domain $B$.
In German, the codomain is referred to as the \emph{Zielmenge} or \emph{Wertevorrat}.
The sense is the same: $\Delta_N(\zeta)$ is the set of values into which the analytic continuation is constrained to fall.
}
The Wertevorrat $\Delta_N(\zeta)$ rigorously contains the full infinite family of all possible analytic continuations at each point $\zeta\in\disk$.

\subsection{The Pick Criterion}

Given distinct $\zeta_1, \zeta_2, \dots, \zeta_n \in \disk$ and any $w_1, w_2, \dots, w_n \in \disk$, there exists a function in the Schur class $f \in \mathscr{S}$ which interpolates the points $f(\zeta_i) = w_i$ for all $i$ if and only if the Pick matrix 
\begin{align}
    \left[
    \frac
    {1-w_i \overline{w}_j}
    {1-\zeta_i\overline{\zeta}_j}
    \right]_{1\leq i, j\leq n}
\end{align}
is positive semidefinite
~\cite{Nicolau2016,Nevalinna1919,Nevalinna1929,Pick1915,BlaschkeBook,PickInterpolationBook}.
If this criterion fails to hold, $f\notin \mathscr{S}$.
For problems of physical interest, we expect the Pick criterion always to hold, at least if the Euclidean data are specified with sufficient precision.
In practice, the Pick criterion may not be satisfied by noisy Monte Carlo data.
This observation was also made in Refs.~\cite{PhysRevLett.126.056402,PhysRevB.104.165111}.
In this case, the formal bounding guarantees may not hold.
However, we expect the Wertevorrat will likely still provide quantitatively useful guidance.
A promising idea is to apply a numerical ``projection," in the same spirit at SVD cuts or shrinkage in standard least-squares fitting, to restore the Pick criterion.
Conceivably, this step could be done in such a way the ``projection" data would be statistically consistent with the original data.
Important work on robust spectral reconstructions in the presence of statistical noise has recently been given in Ref.~\cite{Huang:2022qsb}.

\subsection{Results Mapped to the Upper Half-plane}

For physical interpretation, it now remains to translate the results of the preceding section back to the upper half-plane.
The interpolating functions are readily mapped back by inverting \cref{eq:GRplus_disk,eq:GRminus_disk}:
\begin{align}
\begin{split}
    &G_{+}: \C^+ \to \C^+\\
    &G_{+}(z) = (C^{-1} \circ \mathcal{G}_{+} \circ C)(z)
\end{split}\\
\begin{split}
    &G_{-}: \C^+ \to \C\setminus\R^-\\
    &G_{-}(z) = (\widetilde{C}^{-1} \circ \mathcal{G}_{-} \circ C)(z).
\end{split}
\end{align}
Similarly, the Wertevorr{\"a}te in the upper half-plane are given by
\begin{align}
    \mathcal{D}_N(z) \equiv 
    \begin{cases}
    (C^{-1} \circ \Delta_N \circ C)(z) \text{ for } G_{+}\\
    (\widetilde{C}^{-1} \circ \Delta_N \circ C)(z) \text{ for } G_{-}
    \end{cases}
    z\in\C^+.
    \label{eq:wertevorrat_upper_half_plane}
\end{align}
Using the interpolating function, the smeared spectral density \cref{eq:rhoeps_equals_ImGR} can be computed via
\begin{align}
\rho^\epsilon_{\pm}(\omega)
= \frac{1}{\pi} \imag G_{\pm}(\omega + i \epsilon),
\end{align}
using the interpolating function for $G_{\pm}(z)$.
Since $\Delta_N(\zeta)$ is an open set, so too is $\mathcal{D}_N(z)$, by invariance of domain.
Let $\partial \mathcal{D}_N(z)$ denote its boundary.
The uncertainty in the smeared spectral density is given by the full width of the imaginary part of the boundary of the Wertevorrat,
\begin{align}
    \begin{split}
    \delta \rho^\epsilon_{\pm} (\omega) 
    =& \frac{1}{\pi} \left[\max \imag \partial \mathcal{D}_N(\omega + i \epsilon) \right.\\
    &\left.- \min \imag \partial \mathcal{D}_N(\omega + i \epsilon) \right].
    \end{split}
    \label{eq:error_bound}
\end{align}
The setup for $\rho^\epsilon_{\pm} (\omega)$ and $\delta \rho^\epsilon_{\pm} (\omega)$ is summarized in \cref{fig:wertevorrat}.

\begin{figure}[t]
    \centering
    \includegraphics[width=0.48\textwidth]{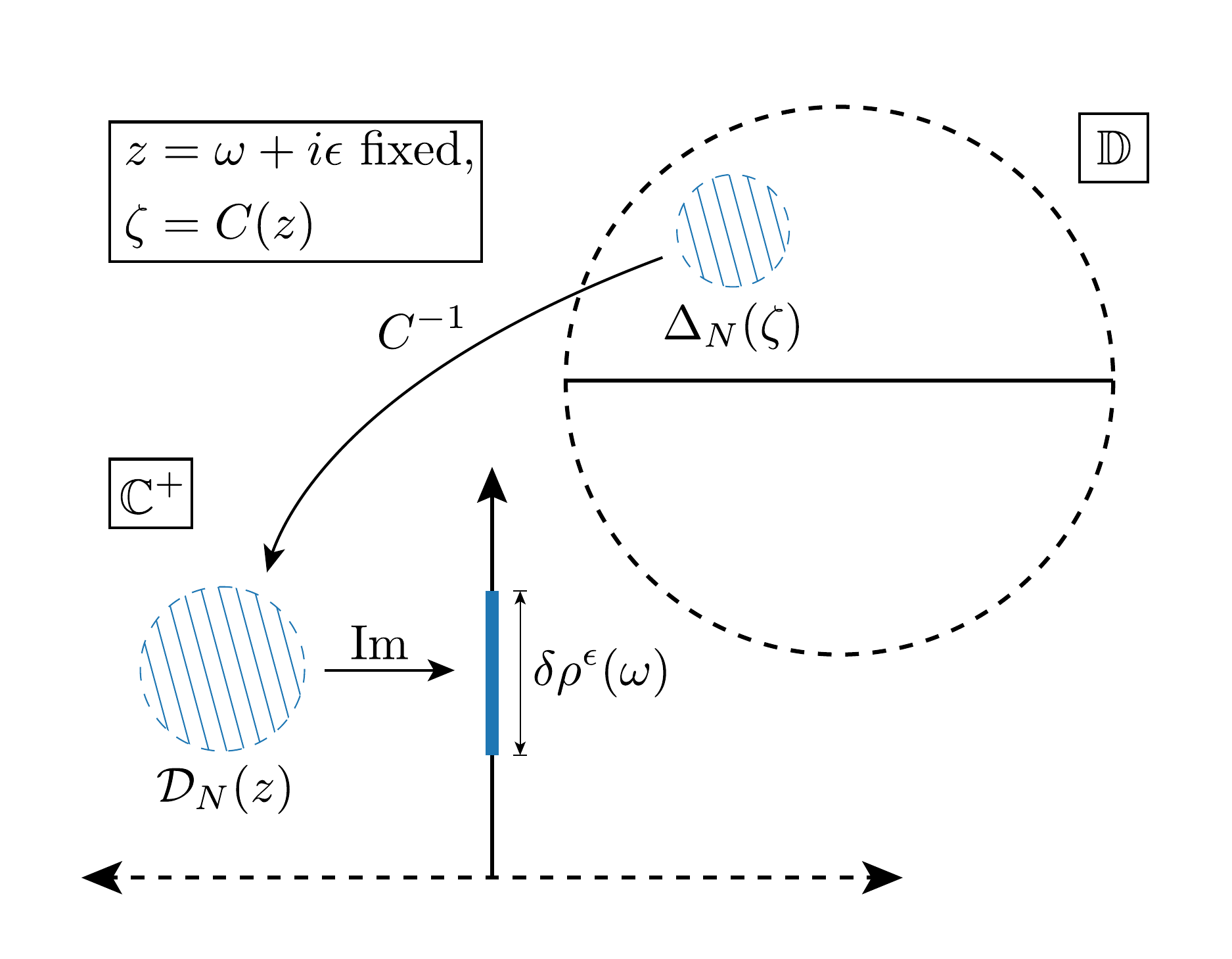}
    \caption{
    Mapping the Wertevorrat to $\C^+$ to evaluate the uncertainty in the smeared spectral density $\delta \rho^\epsilon(\omega)$.
    }
    \label{fig:wertevorrat}
\end{figure}

\Cref{eq:error_bound} has an appealing interpretation for $\epsilon=0$.
As argued above, the Wertevorrat actually fills the full unit disk for boundary values:
\begin{align}
\abs{\zeta}=1 \implies \Delta_N(\zeta) = \overline{\disk}.    
\end{align}
Since $C^{-1}(\overline\disk)=\overline{\C^+}$, the unsmeared spectral density is effectively unconstrained by the analytic continuation procedure:
\begin{align}
    \left. \delta \rho^\epsilon_{\pm}(\omega)\right|_{\epsilon=0} = \infty.
\end{align}
In other words, the ill-posed nature of unsmeared spectral reconstructions is manifest in the Wertevorrat formalism.
However, given Euclidean data $G(i\omega_l)$ at sufficiently many Matsubara frequencies, which amounts to large spatial volumes, $\rho^\epsilon_\pm(\omega)$ can be constrained tightly for $\epsilon > 0$.

\subsection{The Complete Algorithm\label{ssec:complete_algorithm}}

The method to compute smeared spectral densities $\rho^\epsilon(\omega)$ consists of the following steps:
\begin{enumerate}
\item Compute the Fourier coefficients $G_E^{(\ell)}$ via \cref{eq:fourier_coefficients}.\label{step:fourier}
These numbers constitute the Euclidean data.
\item Map Matsubara frequencies and Euclidean data to the unit disk using \cref{eq:cayley} and/or \cref{eq:Ctilde}.
\item Solve the interpolation problem by computing the Nevanlinna coefficients and then compute the boundary of the Wertevorrat using \cref{eq:wertevorrat_boundary}.
\item Map the Wertevorrat from the disk back to the upper half-plane using \cref{eq:wertevorrat_upper_half_plane}.
\item Evaluate the space of smeared spectral densities using \cref{eq:error_bound}.
The true smeared spectral function is rigorously bounded within this space.
\end{enumerate}

Aside from incidental details, the algorithm is the same for both fermionic and bosonic quantities. 
Indeed, apart from the choice of conformal maps to transform the problem to the unit disk, all of the arguments from complex analysis are identical.
Operationally, the two cases are distinguished by the choice of whether to compute Fourier coefficients using the bosonic or fermionic Matsubara frequencies in \cref{step:fourier} above.
This observation has a useful practical consequence.

Recall from \cref{eq:rho_fv_zerotemp} that, for a given set of energies and overlap factors, the zero-temperature spectral densities satisfy 
\begin{align}
\rho_{+}(\omega) = \sgn(\omega) \rho_{-}(\omega).
\end{align}
In particular, this implies that the smeared spectral densities converge to the same numerical result,
\begin{align}
\lim_{\epsilon\to 0}\rho^\epsilon_{+}(\omega) = \sgn(\omega) \lim_{\epsilon\to 0}\rho^\epsilon_{-}(\omega).\label{eq:boson_fermion_limit}
\end{align}
In practice, this means that both the bosonic and fermionic methods define valid smeared spectral densities for zero-temperature Green functions in the sense of \cref{eq:ordered_limit} as $\epsilon \to 0$.
The choice of which one to use when analyzing zero-temperature data therefore becomes a question of convenience and expediency.
In practice, it will likely prove useful to analyze \emph{both} and construct the joint limit as $\epsilon$ approaches zero.
A numerical example of this idea is given below.

Another way to see that $\rho^\epsilon_\pm(\omega)$ both define valid smeared spectral functions at zero temperature, in the sense of \cref{eq:ordered_limit}, is to note that experimental data can be smeared equally well to give $\rho^\epsilon_{+}(\omega)$ or $\rho^\epsilon_{-}(\omega)$, which can be compared to the result of spectral reconstruction.

With all the technical pieces now assembled, it is useful to highlight some similarities and important differences with recent work in the literature.
The first difference has to do with the mapping of the Green functions to the unit disk. 
In Refs~\cite{PhysRevLett.126.056402,PhysRevB.104.165111}, the insightful observation was made that fermionic Green functions have the Nevanlinna property ($\imag G_+(z) > 0$ for all $z\in C^+$), which allowed the problem to be mapped to $\disk$ with the Cayley transform in \cref{eq:cayley}.
The present treatment maps fermionic Green functions to the disk in precisely the same way.
A new observation in the present work is that essentially the same idea holds for bosonic Green functions, provided the conformal map of \cref{eq:Ctilde} is used.
In this sense, our perspective is that the Nevalinna property was a sufficient, but not necessary, condition to enable the use of the interpolation techniques for the unit disk in Ref.~\cite{PhysRevLett.126.056402,PhysRevB.104.165111} and \cref{sec:interpolation}.
Our treatment of bosonic Green functions differs from that of Ref.~\cite{Nogaki:2023mut}, where an auxiliary fermionic problem was constructed using the ``hyperbolic tangent trick"~\cite{Meyer:2007ic,Itou:2020azb} to enable the use of the Cayley transform.
The hyperbolic tanget trick appears to be related to the fluctuation dissipation theorem (cf. Eq.~(39) in Ref.~\cite{Meyer:2011gj}) and holds in the unsmeared limit $\epsilon \to 0$.
It would be interesting to explore its generalization to generic finite $\epsilon$ as well as the combination with error bounds from the Wertevorrat.
Since fermionic spectral densities are positive for all $\omega \in \R$, reconstructions using the method of Ref.~\cite{Nogaki:2023mut} may have attractive positivity features compared to the method in the present work.

A second important difference is related to how results are mapped back to the upper half-plane.
Stated briefly, the algorithm in Refs.~\cite{PhysRevLett.126.056402,PhysRevB.104.165111,Nogaki:2023mut,nogaki2023nevanlinnajl}
yields a single value for the smeared spectral function $\rho^\epsilon(\omega)$ for each $z=\omega +i \epsilon$ (essentially via \cref{eq:rhoeps_equals_ImGR}), while the algorithm in the present work yields a bounding envelope via \cref{eq:error_bound}.
In more detail, the difference stems from the following observations.
As in Refs.~\cite{PhysRevLett.126.056402,PhysRevB.104.165111,Nogaki:2023mut,nogaki2023nevanlinnajl}, 
the fundamental interpolation function on $\disk$ is given by an equation like \cref{eq:interpolating_family} which includes an undetermined function $f_N(\zeta)$ representing the freedom to include additional input data.
In Refs.~\cite{PhysRevLett.126.056402,PhysRevB.104.165111,Nogaki:2023mut,nogaki2023nevanlinnajl}, this freedom is exploited in a smoothing step, which removes spurious oscillations from the resulting spectral function.
The smoothing step constrains the spectral function to minimize a certain convex functional enforcing the normalization sum rule ($\int d\omega \rho(\omega) = 1$) and a notion of smoothness.
To the best of our knowledge, no rigorous field-theoretic justification exists for the precise form of the smoothing criterion.
The smoothed spectral functions certainly appear more physical and do seem to agree well with known results in numerical test cases given in Refs.~\cite{PhysRevLett.126.056402,PhysRevB.104.165111,Nogaki:2023mut,nogaki2023nevanlinnajl}.
However, the perspective of the present work is that the smoothing step introduces an uncontrolled systematic uncertainty.
For applications of interest in lattice QCD, this systematic uncertainty must be quantified.
The present work takes a conservative stance, using the Wertevorrat~\cite{Nevalinna1919,Nevalinna1929,Nicolau2016} to characterize the full space of possible solutions.
The Wertevorrat seems to have been long understood in the mathematical community.
To the best of our knowledge, the present work is the first to recognize its role in bounding uncertainty in analytic continuation in field theory problems.
It would be interesting to combine additional physical constraints (e.g. regarding the moments of the spectral functions as mentioned in Refs.~\cite{PhysRevLett.126.056402,PhysRevB.104.165111} and implemented in Ref.~\cite{nogaki2023nevanlinnajl}) with the novel interpretation of the Wertevorrat.

\section{Numerical calculation of 
\texorpdfstring{$G(i\omega)$}{}
\label{sec:numerical_GR}
}

A few words are in order about how to compute $G(i\omega_l)$ given numerical data for $\mathscr{G}_E(\tau)$.
As shown in \cref{eq:fourier_coefficients}, the desired frequency-space data are the Fourier coefficients of the Euclidean Green function.
Since lattice QCD calculations provide Euclidean data for $\mathscr{G}_E(\tau)$ at the discrete times $\tau\in\{0, 1, \dots \beta-1\}$, it is tempting to compute the Fourier coefficients with the discrete Fourier transform,
\begin{align}
\int_0^\beta d\tau\, e^{i \omega_\ell \tau} \mathscr{G}_E(\tau)
&\stackrel{?}{\approx} \sum_{t=0}^{\beta-1} e^{i \omega_\ell \tau} \mathscr{G}_E(\tau).
\end{align}
For the sake of illustration, let us carry out the bosonic Matsubara sum for a unit-amplitude Green function saturated by a single state, $\mathscr{G}_E(\tau)=e^{-m\tau} + e^{-m(\beta-\tau)}$.
In this case, one readily finds
\begin{align}
\sum_{\tau=0}^{\beta-1} e^{i \omega_\ell \tau} \mathscr{G}_E(\tau)
&= \frac{\left(1 - e^{-\beta m}\right) \sinh(ma)}{\cosh(ma) - \cos( \omega_la)},
\label{eq:dft}
\end{align}
where the implicit dependence on the lattice scale has been made explicit on the right-hand side.
Trying to interpret the right-hand side as $G_+(i\omega_l)$ is immediately problematic.
Due to the periodicity of the denominator, the right-hand side is not analytic in the upper-half plane.
Indeed, besides the expected pair of poles on the real axis, spurious poles also appear in the upper half-plane with offsets at $\pm ma + 2\pi i \Z$.
Nevertheless, the continuum limit has the correct analytic structure,
\begin{align}
\lim_{a\to0} a \frac{\left(1 - e^{-\beta m}\right) \sinh(ma)}{\cosh(ma) - \cos( \omega_la)} = \frac{2 m}{m^2 + \omega^2}.
\end{align}

\begin{figure}[t]
    \centering
    \includegraphics[width=0.48\textwidth]{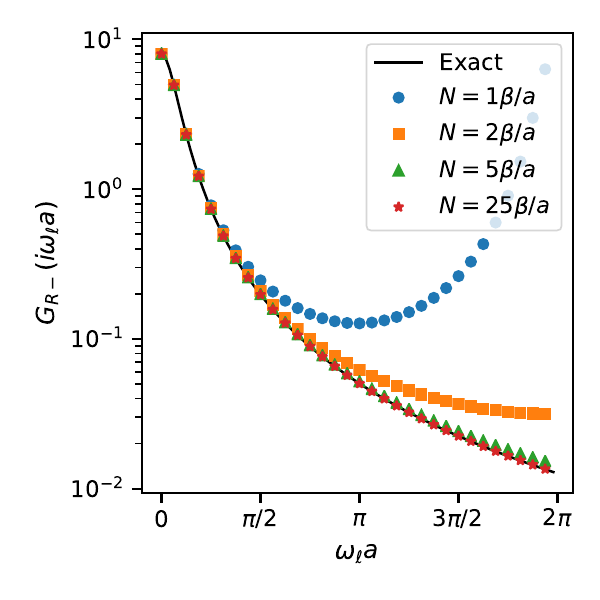}
    \caption{
    Numerical evaluation of the Fourier coefficients $G_E^{(\ell)}$ using different interpolation densities, using a Euclidean correlator $G_E(\tau)=e^{-m\tau} + e^{-m(\beta-\tau)}$ with $\beta=96$ and $m=0.25$.
    To avoid many overlapping points, results are shown for every third Matsubara frequency.
    The different markers show different interpolation refinements.
    The blue circles at the top are the result of the discrete Fourier transform of \cref{eq:dft} and exhibit the expected pathology for large $\omega_\ell$.
    For an interpolation refinement of $25\beta/a$ (red stars), all Fourier coefficients agree with the exact continuum results with sub-percent precision.
    }
    \label{fig:fourier_components}
\end{figure}

This observation immediately suggests a solution, namely, constructing a better approximation to the integral of \cref{eq:fourier_coefficients}.
Since $\mathscr{G}_E(t)$ is smooth and monotonic for $t\in\{0,\beta/2\}$, it can safely be interpolated, and this interpolation can be used to evaluate \cref{eq:fourier_coefficients} numerically on a finer grid.
In practice, it is advantageous to construct an interpolation of $\log \mathscr{G}_E(t)$, e.g., with a simple polynomial spline, since the logarithm is slowly varying.
For example, \cref{fig:fourier_components} shows the result of this procedure for a range of grid refinements with $N\in\beta/a\times\{1,2,5,25\}$ total points.
By interpolating, continuum-like values for $G_\pm(i\omega_l)$ can be obtained, at least for low frequencies.
For a given lattice spacing, estimating the high-frequency 
components ($\omega_l a \gg 1$)
remains difficult from a practical perspective, since essentially all information about these components will have already decayed away by the first Euclidean time.

\section{Numerical Examples\label{sec:examples}}

This section presents numerical examples of the algorithm described in \cref{ssec:complete_algorithm}.
In each case, the exact spectral density will specified.
For simplicity, the zero-temperature limit will be used.
It will be convenient to express the spectral density as a sum of arbitrarily close delta functions via \cref{eq:rho_fv_zerotemp}.
When the spectral density is evaluated numerically on a fine, uniform grid in energy with spacing $\Delta E$, the spectral weights are given by
\begin{align}
|A_{0n}|^2 = \int_{E_n}^{E_n+\Delta E} d\omega\,\rho(\omega),
\end{align}
where $E_n$ is the $n{\rm th}$ energy in the grid.
The retarded Green function evaluated at the Matsubara frequencies then follows directly from \cref{eq:GR_finite_volume}.
These values for $G_E^{(\ell)}=G_{\pm}(i \omega_\ell)$ provide the starting data for the numerical examples.

\subsection{Discrete Spectral Features: Isolated Poles}

\begin{figure}[t]
    \centering
    \includegraphics[width=0.48\textwidth]{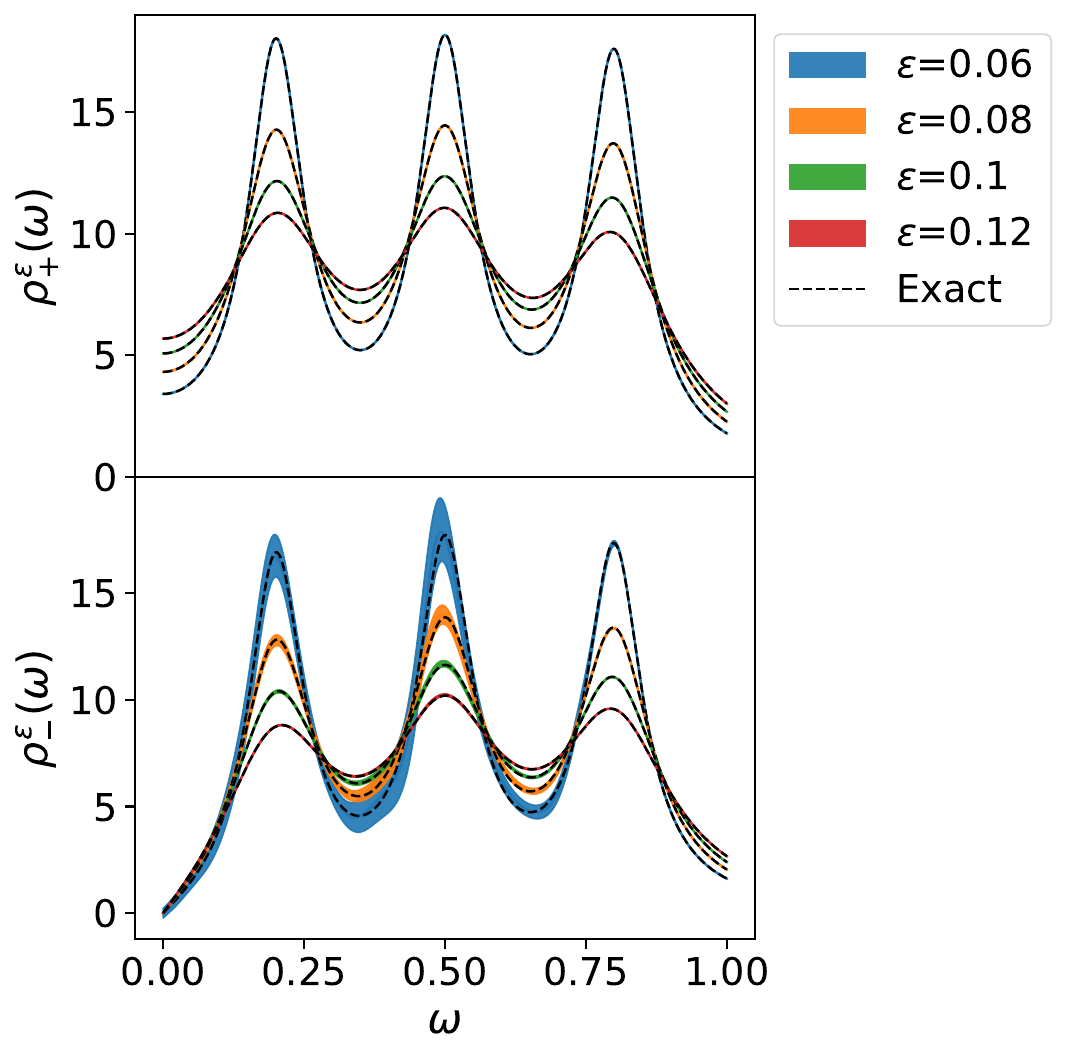}
    \caption{
    Smeared spectral reconstructions of spectral density with delta functions at $\omega \in\{0.2, 0.5, 0.8\}$ and with $\beta=64$.
    The upper (lower) panel shows fermionic (bosonic) reconstructions for several smearing widths.
    The dotted lines show the location of the exact result for each smearing width.
    The exact results lies within bounding envelope of the Wertevorrat.
    }
    \label{fig:three_peak_combined}
\end{figure}

Consider a Green function consisting of three isolated poles with masses $m \in \{ 0.2, 0.5, 0.8\}$ and unit residues.
The corresponding spectral function on the positive real axis is
\begin{align}
    \rho(\omega) =
    \delta(\omega - 0.2)
    +\delta(\omega - 0.5)
    +\delta(\omega -0.8).
\end{align}
Using this spectral function as input, Euclidean data were generated along the imaginary-frequency axis at the Matsubara frequencies with $\beta=64$.
The corresponding reconstructions are shown for a variety of smearing widths in~\cref{fig:three_peak_combined}.
The upper (lower) panel shows the fermionic (bosonic) reconstruction.
In all cases, the exact result lies within the rigorous bounding envelope of the Wertevorrat.
In the fermionic case, the width of the bounding envelope is too small to see visibly.

\subsection{Extended Spectral Features: Gaussians}

\begin{figure}[t]
    \centering
    \includegraphics[width=0.48\textwidth]{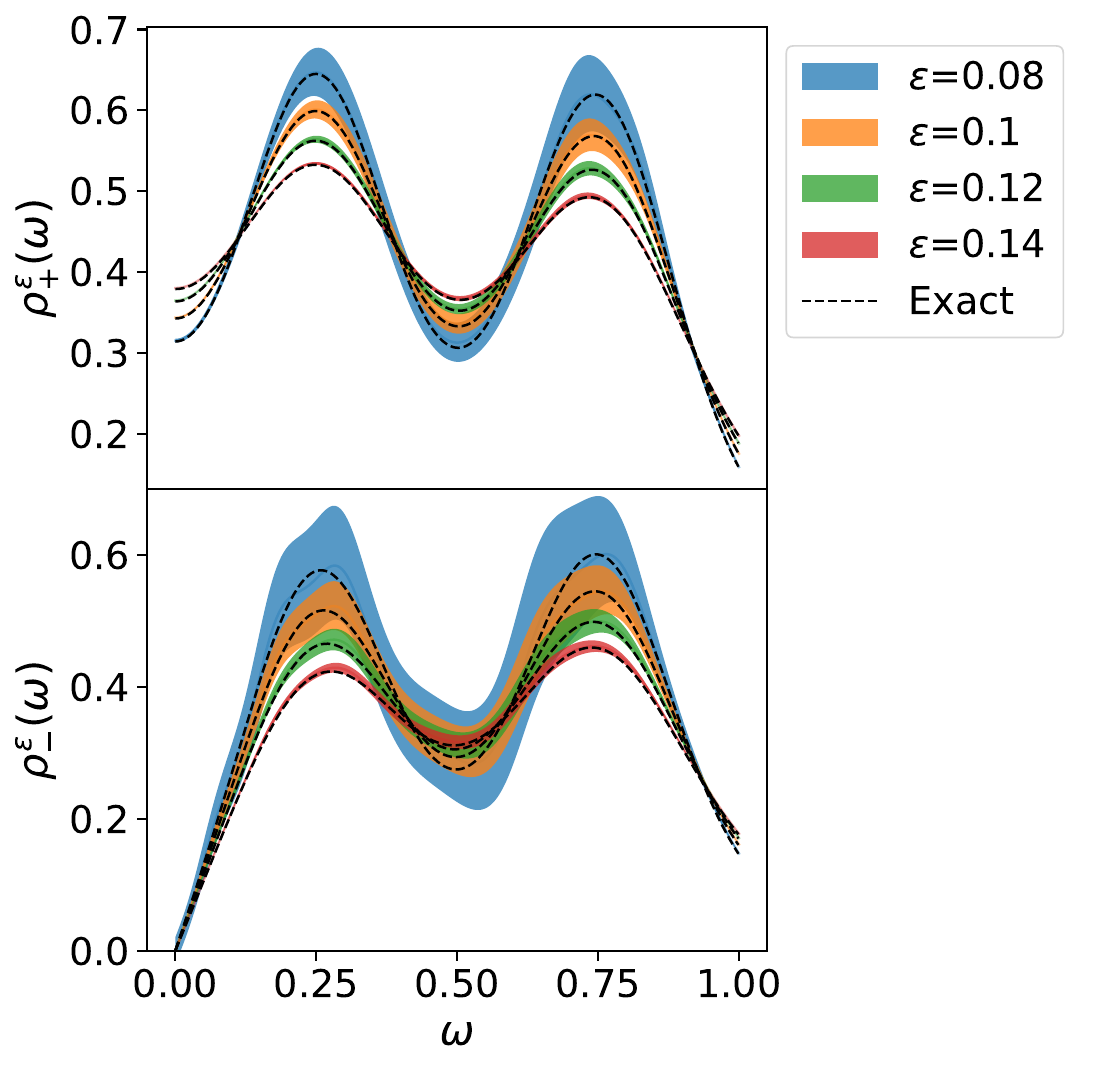}
    \caption{
    Smeared spectral reconstructions of a spectral densitity consisting of a pair of Gaussians, \cref{eq:example_gaussians}.
    Euclidean data were generated with $\beta=48$.
    The upper (lower) panel shows fermionic (bosonic) reconstruction for several smearing widths.
    The dotted lines shows the location of the exact result for each smearing width.
    In all cases, the exact result lies visibly within the bounding envelope of the Wertevorrat.
    }
    \label{fig:gaussian_combined}
\end{figure}

\begin{figure}[t]
    \centering
    \includegraphics[width=0.48\textwidth]{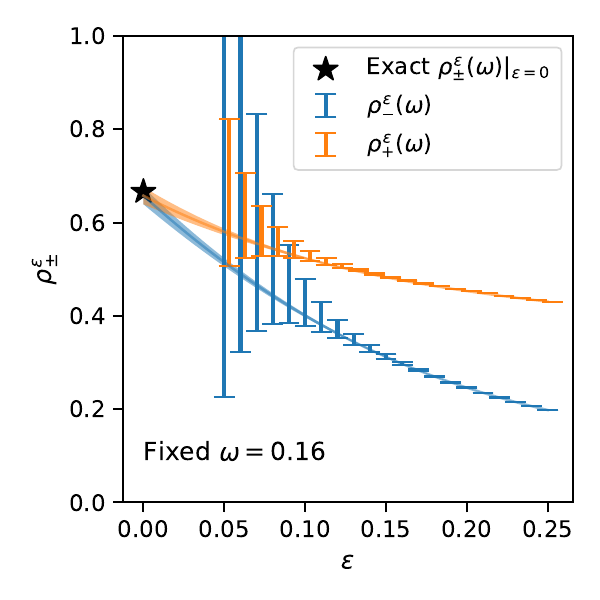}
    \caption{
    Example extrapolation to zero smearing width for the Gaussian example of \cref{eq:example_gaussians} for fixed energy.
    The data points shows reconstructions for $\rho^\epsilon_{\pm}(\omega)$.
    The solid curve shows the result of a polynomial fit imposing \cref{eq:boson_fermion_limit}.
    The black star shows the exact result from \cref{eq:example_gaussians}.
    }
    \label{fig:gaussian_extrapolation}
\end{figure}

Consider a spectral function on the positive real axis consisting of a sum of Gaussians
\begin{align}
    \rho(\omega) = \sum_{i}
    \frac{1}{\sqrt{2\pi}\sigma_i}\exp \left(-\frac{(\omega-\mu_i)^2}{2\sigma_i^2}\right),
    \label{eq:example_gaussians}
\end{align}
where $\mu=\{0.25, 0.75\}$ and $\sigma=\{0.1, 0.1\}$.
Using this spectral function as input, Euclidean data were generated along the imaginary-frequency axis at the Matsubara frequencies with $\beta=48$.

The corresponding reconstructions are shown for a variety of smearing widths in~\cref{fig:gaussian_combined}.
The upper (lower) panel shows the fermionic (bosonic) reconstruction.
In all cases, the exact result lies within the rigorous bounding envelope of the Wertevorrat.
\Cref{fig:gaussian_extrapolation} shows an example of how the bosonic and fermionic reconstruction may be used together to reconstruct their shared liming value of $\rho_{\pm}(\omega)$, cf. \cref{eq:boson_fermion_limit}.

\subsection{The \texorpdfstring{$R-$}~Ratio}

\begin{figure}[t]
    \centering
    \includegraphics[width=0.48\textwidth]{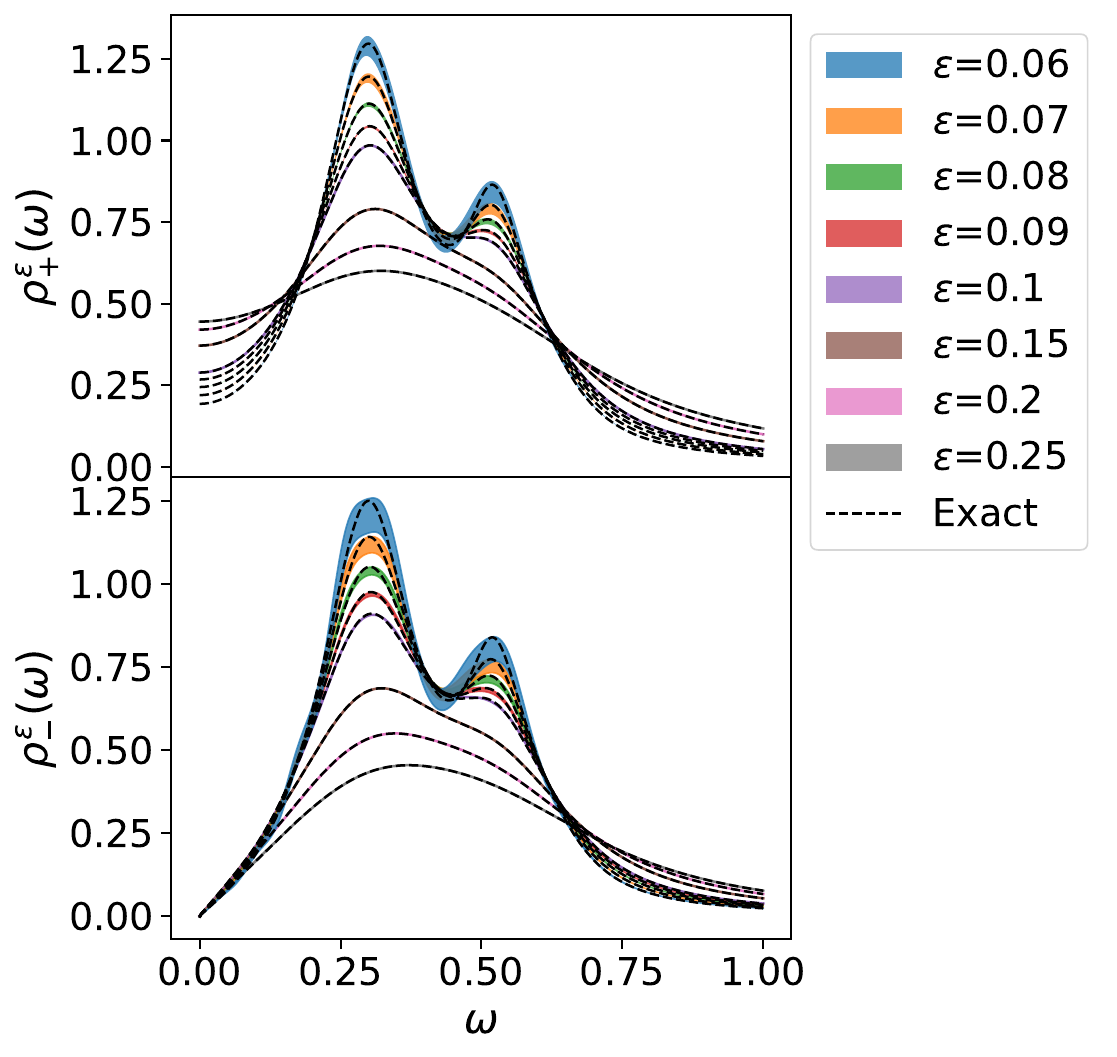}
    \caption{
    Smeared spectral reconstructions of a parameterization of experimental data for the R-ratio.
    Euclidean data were generated with $\beta/a=96$.
    The smearing choices $\epsilon$ are described in the main text.
    In both cases, the exact result lies visibly within the bounding envelope from the Wertevorrat.
    In both cases, the spectral peaks from the $\rho(770)/\omega(782)$ and from the $\phi(1020)$ are clearly visible.
    }
    \label{fig:rratio_combined}
\end{figure}

A parameterization of the experimental data for $R(s)$ is given in Ref.~\cite{Bernecker:2011gh} in terms of Breit-Wigner curves with masses, widths, and amplitudes chosen to match data in the PDG.
The present example uses same parameterization and numerical input values as Ref.~\cite{Bernecker:2011gh}.

Numerical Euclidean data were generated using this parameterization as input for the spectral density with $\beta=96$.
The energy range of the problem was rescaled to lie in the unit interval.
This rescaling places the peak for the $\rho(770)$ at $\omega a \approx 0.25$, which amounts to a lattice spacing of $a\approx 0.07$~fm, a typical cutoff scale appearing in recent calculations of the anomalous magnetic moment of the muon~\cite{Borsanyi:2020mff,chakraborty:2017tqp,
borsanyi:2017zdw,blum:2018mom,giusti:2019xct,shintani:2019wai,FermilabLattice:2019ugu,gerardin:2019rua,Aubin:2019usy,giusti:2019hkz}.

The corresponding reconstructions are shown for a variety of smearing widths in~\cref{fig:rratio_combined}.
The upper (lower) panel shows the fermionic (bosonic) reconstruction.
The two peaks from the $\rho(770)/\omega(782)$ and the $\phi(1020)$ are clearly identified in both reconstructions.
In all cases, the exact result lies within the rigorous bounding envelope of the Wertevorrat.

A reconstruction of $R(s)$ from lattice QCD data using the method of Ref.~\cite{Hansen:2019idp} was recently given in Ref.~\cite{Alexandrou:2022tyn}, where the authors also compared to experimental results including statistical uncertainties.
We hope to conduct a similar test soon using the present method with correlation functions computed using lattice QCD as well as the actual experimental data.

\section{Conclusions\label{sec:conclusions}}

This work has presented a method for numerical analytic continuation, which is closely related to recent work in Refs.~\cite{PhysRevLett.126.056402,PhysRevB.104.165111,Nogaki:2023mut,nogaki2023nevanlinnajl}.
The main application we have in mind is extracting spectral densities from Euclidean data (e.g., computing using lattice QCD).
One of our main insights is that evaluating a Green function in the upper half-plane, \cref{eq:rhoeps_equals_ImGR}, amounts to a smearing prescription in the spirit of \cref{eq:ordered_limit}.
The general formalism is valid for generic (diagonal) thermal Green functions.
As a proof of concept, we presented numerical examples of the method in a few simplified systems.
Important work on robust spectral reconstructions in the presence of statistical noise has recently been given in Ref.~\cite{Huang:2022qsb}.

A key distinguishing feature of our method from recent work in Refs.~\cite{PhysRevLett.126.056402,PhysRevB.104.165111,Nogaki:2023mut,nogaki2023nevanlinnajl}
is the rigorous bounding envelope of the Wertevorrat, which quantifies the systematic uncertainty of the analytic continuation at each point $z$ in the upper half-plane.
The Wertevorrat parameterizes, by explicit construction, the full space of functions consistent with input Euclidean data and with the known causal structure of Green's function in the complex plane.

Another attractive feature of the method is the ease of including additional information, be it experimental or theoretical, to constrain the spectral reconstructions.
For example, when calculating the inclusive structure functions in the resonance region, it might be desirable to constrain the behavior around the elastic scattering peak with lattice QCD data from simpler three-point correlation functions.
Likewise, guidance from perturbation theory may usefully constrain and stabilize reconstructions at high energies.
Regardless of precise source of the constraint, all that's required is to translate the constraints to the upper-half plane and include them numerically in the interpolation problem.
Moreover, the constraining data can be specified wherever it is known most precisely, including just above the real line and \emph{not} necessarily on the imaginary-frequency axis.
As with the general method itself, all constraints are included non-parametrically.
Recent work has also explored the inclusion of constraints, especially moments of the spectral function~\cite{PhysRevLett.126.056402,PhysRevB.104.165111,nogaki2023nevanlinnajl}.

Looking beyond applications in lattice QCD, the problem of computing inverse Laplace transforms arises in many fields.
For instance, in nuclear theory, Green's function Monte Carlo is often used to infer nuclear electroweak response functions from their Laplace transforms, the so-called Euclidean response functions.~\cite{Carlson:2014vla,King:2020wmp}
It would be interesting to explore the application of the ideas explored in this paper to other domains.

\section*{Acknowledgments}
We gratefully acknowledge useful discussions with and comments from
Ryan Abbott,
Tom Blum,
J{\'e}r{\^o}me Charles,
Tom DeGrand,
Will Detmold,
Luchang Jin,
Andreas Kronfeld,
Phiala Shanahan,
Doug Stewart,
and Julian Urban.
We thank Emanuel Gull and Kosuke Nogaki for guidance regarding the beautiful recent work in the condensed matter community.
This material is based upon work supported in part 
by the U.S. Department of Energy, Office of Science under grant Contract Numbers DE-SC0011090 and DE-SC0021006.

\appendix
\section{Technical Material \label{app:disk_bounds}}

\begin{lemma}
Let $w,z \in \C$ with $\overline{w}z\neq 1$.
If $\abs{z}<1$ and $\abs{w}<1$, then $\abs{\frac{w-z}{1-\overline{w}z}} < 1$.
If $\abs{z}=1$ or $\abs{w}=1$, then $\abs{\frac{w-z}{1-\overline{w}z}} = 1$.
\end{lemma}
\begin{proof}
The proof is by calculation. Take the polar decomposition $z=re^{i\varphi}$.
Then
\begin{align}
\left|\frac{w-z}{1-\overline{w}z} \right|
=\left|\frac{w e^{-i\varphi}-r}{1-\overline{we^{-i\varphi}}r} \right|,
\end{align}
so it suffices to take $z=r\in \R$.
This amounts to determining when
\begin{align*}
\frac{w-r}{1-\overline{w}r} \frac{\overline{w}-r}{1-wr} &\leq 1\\
\iff
(w-r)(\overline{w}-r) &\leq (1 - wr)(1 - \overline{w}r)
\end{align*}
Taking $w=\abs{w}e^{i\theta}$, the previous inequality reduces to
\begin{align*}
\abs{w}^2 + r^2 \leq 1 +\abs{w}^2 r^2.
\end{align*}
The inequality is clearly satisfied when both $r < 1$ and $\abs{w} < 1$.
The equality holds when either $r=1$ or $\abs{w}=1$.
\end{proof}
The preceding lemma has immediate consequences for the interpolation problem in Sec.~\ref{sec:interpolation}.
For fixed $z\in\overline{\disk}$, the map $U_{n,z}: \disk \to \disk$ defined by
\begin{align}
    U_{n,z}(w)
    = \frac{w_n^{(n-1)} + b_{z_n}(z) w}{1 + \overline{w_n^{(n-1)}} b_{z_n}(z) w}
\end{align}
is clearly into.
The lemma also shows that $\abs{b_{z_n}(z)} =1$ when $\abs{z}=1$.
Therefore, when $\abs{z}=1$, $U_{n,z}$ is actually onto, i.e., the image of $U_{n,z}$ is the full open disk.
The composition properties of M{\"obius} transformations show, by induction, that the same is true for \cref{eq:Tnz}.

\bibliography{references.bib}

\end{document}